\documentclass[12pt,onecolumn]{IEEEtran}
\usepackage{amssymb}
\usepackage[cmex10]{amsmath}
\usepackage{amsfonts}
\usepackage[latin1]{inputenc}
\usepackage{makeidx}         % allows index generation
\usepackage{multicol}        % used for the two-column index
\usepackage[bottom]{footmisc}% places footnotes at page bottom
\usepackage{graphicx}
\usepackage{url}
\usepackage{color}
\newtheorem{theorem}{Theorem}[section]

\newtheorem{proposition}[theorem]{Proposition}

\newtheorem{remark}[theorem]{Remark}

\newfont{\BB}{msbm10 scaled\magstep1}
\newfont{\bb}{msbm8}

\def\R{\mbox{\BB R}}
\def\E{\mbox{\BB E}}
\def\D{\mbox{\BB D}}

\def\log{{\rm ln}\ }

\newcommand{\x}{x} %{\xi}
\newcommand{\xx}{x} %{\xi}
\newcommand{\vv}{v} %{\xi}

\newcommand{\tstart}{{0}} %{\xi}
\newcommand{\tend}{{1}} %{\xi}

\begin{document}

\title{On the relation between {optimal transport and Schr\"{o}dinger bridges:} A stochastic control viewpoint}
\author{Yongxin Chen, Tryphon Georgiou and Michele Pavon
\thanks{Y. Chen and T. Georgiou are with the Department of Electrical and Computer Engineering,
University of Minnesota,
200 Union street S.E. Minneapolis, Minnesota MN 55455, U.S.A. {\tt\small chen2468@umn.edu}, {\tt\small tryphon@umn.edu}}
\thanks{Michele Pavon is with the Dipartimento di Matematica,
Universit\`a di Padova, via Trieste 63, 35121 Padova, Italy {\tt\small pavon@math.unipd.it}}}
\markboth{December 14, 2014}{}
\maketitle
\begin{abstract}
We take a new look at the relation between the optimal transport problem and the Schr\"{o}dinger bridge problem from the stochastic control perspective. We show that the connections are richer and deeper than described in existing literature. In particular: a) We give an elementary derivation of the Benamou-Brenier fluid dynamics version of the optimal transport problem; b) We provide a new fluid dynamics version of the Schr\"{o}dinger bridge problem; c) We observe that the latter provides an important connection with optimal transport without zero noise limits; d) We propose and solve a fluid dynamic version of optimal transport with {\em prior\/}; { e) We can then view optimal transport with prior as the zero noise limit of Schr\"{o}dinger bridges when the prior is any Markovian evolution. In particular, we work out the Gaussian case. A numerical example of the latter convergence involving Brownian particles is also  provided. }
\end{abstract}

\section{Introduction}
We discuss two problems of very different beginning. Optimal mass transport (OMT) originates in the work of Gaspar Monge in 1781 \cite{Monge} and seeks a transport plan that corresponds in an optimal way two distributions of equal total mass. The cost penalizes the distance that mass is transported to ensure exact correspondence. Likewise, data for Erwin Schr\"odinger's 1931/32 bridge problem  \cite{S1,S2} are again two distributions of equal total mass, in fact, probability distributions. Here however, these represent densities of diffusive particles at two points in time and the problem seeks the most likely path that establishes a correspondence between the two. A rich relationship between the two problems emerges in the case where the transport cost is quadratic in the distance, and in fact, the problem of OMT emerges as the limit of Schr\"odinger bridges as the diffusivity tends to zero. The parallel treatment of both problems highlights the time-symmetry of both problems and points of contact between stochastic optimal control and information theoretic concepts.

Historically, the modern formulation of optimal mass transport is due to Leonid Kantorovich \cite{Kantorovich} and the subject has been the focus of renewed and increased interest because of its relevance in a wide range of fields including economics, physics, engineering, and probability \cite{RachevR,Vil,Villani_oldnew}. In fact, Kantorovich's contributions and their impact to resource allocation was recognized with the Nobel Prize in Economics in 1975 while in the past twenty years contributions by Ambrosio,  Benamou, Brenier, McCann, Cullen, Gangbo, Kinderlehrer, Lott, Otto, Rachev, R\"uschendorf, Tannenbaum, Villani, and many others have launched a new fast developing phase, see e.g., \cite{GangboMcCann,Otto,BB,AGS,Vil,Villani_oldnew,NGT}.
On the other hand, the Schr\"{o}dinger bridge problem  \cite{S1,S2}, has been the subject of strong but intermittent interest by mostly probabilists, physicists, and quantum theorists.
Early important contributions were due to Fortet, Beurling, Jamison and F\"{o}llmer \cite{For,Beu,Jam,F2}, see \cite{W} for a survey. Renewed interest was sparked in the past twenty years after a close relationship to stochastic control was recognized \cite{DP,DPP,PW} and a similarly fast developing phase is underway, see the semi-expository paper \cite{leo} and \cite{MT,PT2,leo2,GP} for other recent contributions.

Besides the intrinsic importance of optimal mass transport to the geometry of spaces and the multitude of applications, a significant impetus for some recent work has been the need for effective computation \cite{BB}, \cite{AHT} which is often challenging. Likewise,
excepting special cases \cite{FH,FHS}, the computation of the optimal stochastic control for the Schr\"{o}dinger bridge problem is challenging, as it amounts to two partial differential equations nonlinearly coupled through their boundary values \cite{W}. Only very recently  implementable forms have become available for corresponding linear stochastic systems  \cite{CG2014,CGP,CGP3} and for versions of the problem involving Markov chains and Kraus maps of statistical quantum mechanics \cite{GP}; see also \cite{CGP4} which deals with the Schr\"{o}dinger bridge problem with finite or infinite horizon for a system of nonlinear stochastic oscillators.

The aim of the present paper is to elucidate some of the connections between optimal mass transport and Schr\"odinger bridges { thereby extending both theories}. We follow in the footsteps of L\'eonard \cite{leo2,leo}, who investigated their relation, and of Mikami and Thieullen \cite{Mik, mt, MT} who employed stochastic control and Schr\"{o}dinger bridges to solve the optimal transport problem.
%Yet, we feel that the picture of the relationship between these two areas is not complete. %Indeed, \cite{leo} is not concerned with optimal control, whereas \cite{Mik,MT} concentrates on duality and zero noise limits.
This paper may then be seen to complement the results in these  papers by providing a unifying view of the relationship between these two problems via optimal control. In particular, we give an elementary derivation of the Benamou-Brenier fluid dynamics version of the Monge-Kantorovich problem. We also provide a time-symmetric fluid dynamic version of the Schr\"{o}dinger bridge problem different from \cite[Section 4]{leo}; it underscores that an important connection with optimal transport exists even without zero noise limits. We then formulate and solve a fluid dynamic version of the optimal transport problem with {\em prior}.  { This allows us to study zero noise limits of Schr\"{o}dinger bridges when the prior is any Markovian evolution. In particular, employing our results of \cite{CGP}, we study the case when the prior evolution is a Gauss-Markov process.

The outline of the paper is as follows: In Section \ref{BB}, we derive the Benamou-Brenier version of the OMT problem. In Section \ref{background}, we provide some background on the classical Schr\"{o}dinger bridge problem. Section \ref{SC} is devoted to characterizing the optimal forward and backward drift in the bridge problem. In Section \ref{TSF}, we give a control time-symmetric formulation of the Schr\"{o}dinger bridge problem. This leads, in the following Section \ref{FDFS}, to a new fluid dynamic formulation of the bridge problem. Section \ref{OMTprior} is dedicated to the optimal mass transfer problem with prior. In Section \ref{gaussian}, we investigate the zero noise limit when the prior is Gaussian. The paper concludes with two examples. In Section \ref{shifting}, we discuss the zero noise limit when the prior is Wiener measure and the goal is shifting the mean of a normal distribution. Finally, in Section \ref{NE}, we provide a numerical two-dimensional example of overdamped Brownian particles. In the zero noise limit, we obtain the solution of the corresponding OMT problem with prior.}

\section{Optimal mass transport as a stochastic control problem}\label{BB}
\subsection{The Monge-Kantorovich problem}
Given two distributions $\mu,\nu$ on ${\mathbb R}^n$ having equal total mass, the original formulation due to G.\ Monge sought to identify a transport (measurable) map $T$ from ${\mathbb R}^n\to{\mathbb R}^n$ so that the push-forward $T\sharp \mu$ is equal to $\nu$, in the sense that $\nu(\cdot)=\mu(T^{-1}(\cdot))$, while the cost of transportation $\int c(x,T(x))\mu(dx)$ is minimal. Here, $c(x,y)$ represents the transference cost from point $x$ to point $y$ and for the purposes of the present it will be $c(x,y)=\frac{1}{2}\|x-y\|^2$.

The dependence of the cost of transportation on $T$ is highly nonlinear which complicated early analyses of the problem. Thus, it was not until Kantorovich's relaxed formulation in 1942 that the Monge's problem received a definitive solution. In this, instead of the transport map one seeks a joint distribution $\Pi(\mu,\nu)$ on the product space $\R^n\times\R^n$, refered to as a ``coupling" between $\mu$ and $\nu$, so that the marginals along the two coordinate directions coincide with $\mu$ and $\nu$ respectively. Thence, one seeks to determine
\begin{equation}\label{OptTrans}
\inf_{\pi\in\Pi(\mu,\nu)}\int_{\R^n\times\R^n}\frac{1}{2}\|x-y\|^2d\pi(x,y)
\end{equation}
%where $\Pi(\mu,\nu)$, called the ``couplings" between $\mu$ and $\nu$, denotes the collection of all probability measures $\pi(x,y)$ on $\R^n\times\R^n$ with specified marginals $\mu(x)$ and $\nu(y)$.
In case an optimal transport map exists, the optimal coupling has support on the graph of this map, see \cite{Vil}. Herein, we consider this relaxed Kantorovich formulation.
We wish first to give next an elementary derivation of the fact that Problem \ref{OptTrans} can be turned into a stochastic control problem as stated in \cite[formula (1.6)]{MT} and then, to derive an alternative ``fluid-dynamic'' formulation due to Benamou-Brenier. We strive for clarity rather than generality. {In particular, we (tacitly) assume throughout the paper that $\mu$ does not give mass to sets of dimension $\le n-1$. Then, by Brenier's theorem \cite{Vil}, there exists a unique optimal transport plan (Kantorovich) induced by a map (Monge) which is the gradient of a convex function.}

\subsection{A stochastic control formulation}
As customary, let us start by observing that
\begin{equation}\label{calvar}
\frac{1}{2}\|x-y\|^2=\inf_{\x\in\mathcal X_{xy}}\int_{\tstart}^{\tend}\frac{1}{2}\|\dot{\x}\|^2dt
\end{equation}
where $\mathcal X_{xy}$ is the family of $C^1([\tstart,\tend];\R^n)$ paths with $\x(\tstart)=x$ and $\x(\tend)=y$. Let
$$\x^*(t)= (\tend-t)x+ty
$$
be the solution of (\ref{calvar}), namely the straight line joining $x$ and $y$. Since $\x^*(t)$ is a Euclidean geodesic, any probabilistic average of the lengths of $C^1$ trajectories starting at $x$ at time $\tstart$ and ending in $y$ at time $\tend$ gives necessarily a higher value. Thus, the probability measure on $C^1([\tstart,\tend];\R^n)$ concentrated on the path $\{\x^*(t);\tstart\le t\le \tend\}$ solves the following problem
\begin{equation}\label{stoch}
\inf_{P_{xy}\in\D^1(\delta_{x},\delta_{y})}\E_{P_{xy}}\left\{\int_{\tstart}^{\tend}\frac{1}{2}\|\dot{\x}\|^2dt\right\},
\end{equation}
where $\D^1(\delta_{x},\delta_{y})$ are the probability measures on $C^1([\tstart,\tend];\R^n)$ whose initial and final one-time marginals are Dirac's deltas concentrated at $x$ and $y$, respectively. Since (\ref{stoch}) provides us with yet another representation for $\frac{1}{2}\|x-y\|^2$, in view of (\ref{OptTrans}), we also get that
\begin{eqnarray}\label{OptTrans2}
\inf_{\pi\in\Pi(\mu,\nu)}\int_{\R^n\times\R^n}\frac{1}{2}\|x-y\|^2d\pi(x,y)&=& \\
&&\hspace*{-2cm}
\inf_{\pi\in\Pi(\mu,\nu)}\int_{\R^n\times\R^n}\inf_{P_{xy}\in\D^1(\delta_{x},\delta_{y})}\E_{P_{xy}}\left\{\int_{\tstart}^{\tend}\frac{1}{2}\|\dot{\x}\|^2dt\right\}d\pi(x,y)\nonumber
\end{eqnarray}
Now observe that if $P_{xy}\in\D^1(\delta_{x},\delta_{y})$ and $\pi\in\Pi(\mu,\nu)$ then 
$$P=\int_{\R^n\times\R^n}P_{xy}d\pi(x,y)$$
is a probability measure in $\D^1(\mu,\nu)$, namely a measure on $C^1([\tstart,\tend];\R^n)$ whose one-time marginal at $t_{0}$ and $t_{1}$ are specified to be $\mu$ and $\nu$, respectively. Conversely, the disintegration of any measure $P\in\D^1(\mu,\nu)$ with respect to the initial and final positions yields $P_{xy}\in\D^1(\delta_{x},\delta_{y})$ and $\pi\in\Pi(\mu,\nu)$. Thus we get that the original optimal transport problem is equivalent to
\begin{equation}\label{stoch2}
\inf_{P\in\D^1(\mu,\nu)}\E_P\left\{\int_{\tstart}^{\tend}\frac{1}{2}\|\dot{\x}\|^2dt\right\}.
\end{equation}
So far, we have followed \cite[pp.\ 2-3]{leo}.
Instead of the ``particle'' picture, we can also consider the hydrodynamic version of (\ref{calvar}), namely the optimal control problem
\begin{eqnarray}\label{calvarhydro1}
\frac{1}{2}\|x-y\|^2=\inf_{\vv\in\mathcal V_{y}}\int_{\tstart}^{\tend}\frac{1}{2}\|\vv(\x^\vv(t),t)\|^2dt\\
\dot\x^\vv(t)=\vv(\x^\vv(t),t),\quad \x(\tstart)=x,\nonumber
\end{eqnarray}
where the admissible feedback control laws $\vv(\cdot,\cdot)$ in $\mathcal V_y$ are continuous and such that $\x^\vv(\tend)=y$.

Following the same steps as before, we get that the optimal transport problem is equivalent to the following stochastic control problem with atypical boundary constraints
\begin{subequations}\label{eq:stochcontr}
\begin{eqnarray}\label{stochcontr1}&&\inf_{\vv\in\mathcal V}\E\left\{\int_{\tstart}^{\tend}\frac{1}{2}\|\vv(\x^\vv(t),t)\|^2dt\right\}\\&& \dot{\x}^\vv(t)=\vv(\x^\vv(t),t),\quad {\rm a.s.},\quad \x(\tstart)\sim\mu,\quad \x(\tend)\sim\nu.\label{stochcontr2}
\end{eqnarray}
\end{subequations}
Finally suppose $d\mu(x)=\rho_0(x)dx$, $d\nu(y)=\rho_1(y)dy$ and $\x^\vv(t)\sim\rho(t,x)dx$. Then, necessarily, $\rho$ satisfies (weakly) the continuity equation
\begin{equation}\label{continuity}
\frac{\partial \rho}{\partial t}+\nabla\cdot(\vv\rho)=0
\end{equation}
expressing the conservation of probability mass. Moreover,
$$\E\left\{\int_{\tstart}^{\tend}\frac{1}{2}\|\vv(\x^\vv(t),t)\|^2dt\right\}=\int_{\R^n}\int_{\tstart}^{\tend}\frac{1}{2}\|\vv(x,t)\|^2\rho(t,x)dtdx.
$$
Hence (\ref{eq:stochcontr}) turns into the celebrated ``fluid-dynamic'' version of the optimal transport problem due to Benamou and Brenier \cite{BB}:
%which expresses once again as a stochastic control problem in the form
\begin{subequations}\label{eq:BB}
\begin{eqnarray}\label{BB1}&&\inf_{(\rho,v)}\int_{\R^n}\int_{\tstart}^{\tend}\frac{1}{2}\|\vv(x,t)\|^2\rho(t,x)dtdx,\\&&\frac{\partial \rho}{\partial t}+\nabla\cdot(\vv\rho)=0,\label{BB2}\\&& \rho(\tstart,x)=\rho_0(x), \quad \rho(\tend,y)=\rho_1(y).\label{boundary}
\end{eqnarray}\end{subequations}
The variational analysis for (\ref{eq:stochcontr}) or, equivalently, for (\ref{eq:BB})  can be carried out in many different ways. For instance, let $\mathcal P_{\rho_0\rho_1}$ be the family of flows of probability densities $\rho=\{\rho(\cdot,t); \tstart\le t\le \tend\}$ satisfying (\ref{boundary}) and let $\mathcal V$ be the family of continuous feedback control laws $\vv(\cdot,\cdot)$. Consider the unconstrained minimization of the Lagrangian over $\mathcal P_{\rho_0\rho_1}\times\mathcal V$
\begin{equation}\label{lagrangian}
\mathcal L(\rho,v)=\int_{\R^n}\int_{\tstart}^{\tend}\left[\frac{1}{2}\|\vv(x,t)\|^2\rho(t,x)+\lambda(x,t)\left(\frac{\partial \rho}{\partial t}+\nabla\cdot(\vv\rho)\right)\right]dtdx,
\end{equation}
where $\lambda$ is a $C^1$ Lagrange multiplier. Integrating by parts, assuming that limits for $x\rightarrow\infty$ are zero, we get
\begin{equation}\label{lagrangian2}\int_{\R^n}\int_{\tstart}^{\tend}\left[\frac{1}{2}\|\vv(x,t)\|^2+\left(-\frac{\partial \lambda}{\partial t}-\nabla\lambda\cdot \vv)\right)\right]\rho(x,t)dtdx+\int_{\R^n}\left[\lambda(x,\tend)\rho_1(x)-\lambda(x,\tstart)\rho_0(x)\right]dx.
\end{equation}
The last integral is constant over $\mathcal P_{\rho_0\rho_1}$ and can therefore be discarded. We are left to minimize
\begin{equation}\label{lagrangian3}\int_{\R^n}\int_{\tstart}^{\tend}\left[\frac{1}{2}\|\vv(x,t)\|^2+\left(-\frac{\partial \lambda}{\partial t}-\nabla\lambda\cdot \vv)\right)\right]\rho(x,t)dtdx
\end{equation}
over $\mathcal P_{\rho_0\rho_1}\times\mathcal V$.
We consider doing this in two stages, starting from minimization with respect to $\vv$ for a fixed flow of probability densities $\rho=\{\rho(\cdot,t); \tstart\le t\le \tend\}$ in $\mathcal P_{\rho_0\rho_1}$. Pointwise minimization of the integrand at each time $t\in[\tstart,\tend]$
%, namely of
%$$\inf_{v\in\R^n}\left[\frac{1}{2}\|v\|^2-v\cdot\nabla\lambda(x,t)\right]
%$$
gives that
\begin{equation}\label{optcond}
\vv^*_\rho(x,t)=\nabla\lambda(x,t)
\end{equation}
which is continuous. Plugging this form of the optimal control into (\ref{lagrangian3}), we get %the functional of $\rho\in\mathcal X_{\rho_0\rho_1}$,
\begin{equation}
J(\rho)=-\int_{\R^n}\int_{\tstart}^{\tend}\left[\frac{\partial \lambda}{\partial t}+\frac{1}{2}\|\nabla\lambda\|^2\right]\rho(x,t)dtdx.
\end{equation}
In view of this, if $\lambda$ satisfies the Hamilton-Jacobi equation
\begin{equation}\label{HJ}
\frac{\partial \lambda}{\partial t}+\frac{1}{2}\|\nabla\lambda\|^2=0,
\end{equation}
then $J(\rho)$ is identically zero over $\mathcal P_{\rho_0\rho_1}$ and any $\rho\in\mathcal P_{\rho_0\rho_1}$ minimizes the Lagrangian (\ref{lagrangian}) together with the feedback control (\ref{optcond}).
We have therefore established the following \cite{BB}:
\begin{proposition}Let $\rho^*(x,t)$ with $t\in[\tstart,\tend]$ and $x\in {\mathbb R}^n$, satisfy
\begin{equation}\label{optev}
\frac{\partial \rho^*}{\partial t}+\nabla\cdot(\nabla\psi\rho^*)=0, \quad \rho^*(x,\tstart)=\rho_0(x),
\end{equation}
where $\psi$ is a solution of the Hamilton-Jacobi equation
\begin{equation}\label{HJclass}
\frac{\partial \psi}{\partial t}+\frac{1}{2}\|\nabla\psi\|^2=0
\end{equation}
for some boundary condition $\psi(x,\tend)=\psi_1(x)$.
If $\rho^*(x,\tend)=\rho_1(x)$, then the pair $\left(\rho^*,\vv^*\right)$ with $\vv^*(x,t)=\nabla\psi(x,t)$ is a solution of (\ref{eq:BB}).
\end{proposition}
The stochastic nature of the Benamou-Brenier formulation \eqref{eq:BB} stems from the fact that initial and final densities are specified. Accordingly, the above requires solving a two-point boundary value problem and the resulting control dictates the local velocity field. In general, one cannot expect to have a classical solution of (\ref{HJclass}) and has to be content with a viscosity solution \cite{FS}. See \cite{TT} for a recent contribution in the case when only samples of $\rho_0$ and $\rho_1$ are known.

\section{Backround on Schr\"{o}dinger Bridges}\label{background}
\subsection{Finite energy diffusions}
We follow \cite{Jam,F2,W}. Let $\Omega:=C([\tstart,\tend],\R^n)$ denote the family of $n$-dimensional
continuous functions, let
$W_x$ denote Wiener measure on $\Omega$ starting at $x$ at $t=\tstart$, and let
$$W:=\int W_x\,dx
$$
be stationary Wiener measure. Let $\D$ be the family of
distributions on $\Omega$ that are equivalent to $W$. By Girsanov's theorem, under $Q\in\D$, the coordinate process $\xx(t,\omega)=\omega(t)$ admits the representations
\begin{eqnarray}\label{for}
d\xx(t)&=&\beta_+^Q dt+ dw_+(t), \quad \beta_+^Q \;{\rm is}\; {\cal F}_t^+ - {\rm adapted},\\
d\xx(t)&=&\beta_-^Q dt+ dw_-(t), \quad \beta_-^Q \;{\rm is}\; {\cal F}_t^- - {\rm adapted},
\label{bac}\end{eqnarray}
where ${\cal F}_t^+$ and ${\cal F}_t^-$ are $\sigma$- algebras of events observable up to time $t$ and from time $t$ on, respectively, { and $w_-$, $w_+$ are standard $n$-dimensional Wiener processes,} \cite{Foe}. Moreover,
$$Q\left[\int_{\tstart}^{\tend}\|\beta_+^Q\|^2dt <\infty\right]=Q\left[\int_{\tstart}^{\tend}\|\beta_-^Q\|^2dt <\infty\right]=1.
$$
For
$Q,P\in\D$, we define  the {\it relative
entropy } $H(Q,P)$ of $Q$ with respect to $P$ as
$$H(Q,P)=E_Q\left[\log\frac{dQ}{dP}\right].$$
It then follows from Girsanov's theorem that
{
\begin{subequations}
\begin{eqnarray}\label{RE1}
H(Q,P)&=&H(q_0,p_0)+E_Q\left[\int_{\tstart}^{\tend}\frac{1}{2}
\|\beta_+^Q-\beta_+^P\|^2dt\right]\\
&=&H(q_1,p_1)+E_Q
\left[\int_{\tstart}^{\tend}\frac{1}{2}
\|\beta_-^Q-\beta_-^P\|^2dt\right].\label{RE2}
\end{eqnarray}
\end{subequations}
}
Here { $q_0$, $q_1$ are the marginal densities of $Q$ at $\tstart$ and $\tend$, respectively. Similarly, $p_0$, $p_1$ are the marginal densities of $P$. Then, $\beta_+^Q$ and
$\beta_-^Q$ are the
forward and the backward drifts of $Q$, respectively, and similarly for $P$.} The sketch of the proof goes as follows:
\begin{equation}
\begin{split}
\frac{dQ}{dW}&=q_0(\xx(\tstart))\exp\left[\int_{\tstart}^{\tend}\beta_+^Q d\xx-\int_{\tstart}^{\tend}\frac{1}{2}\|\beta_+^Q\|^2dt\right],\quad Q \; {\rm a.s.},\\
\frac{dW}{dP}&=\frac{1}{p_0(\xx(\tstart))}\exp\left[-\int_{\tstart}^{\tend}\beta_+^P d\xx+\int_{\tstart}^{\tend}\frac{1}{2}\|\beta_+^P\|^2dt\right],\quad P \; {\rm a.s.}\Rightarrow Q \; {\rm a.s.}.
\end{split}
\end{equation}
Hence,
\begin{equation}
\begin{split}
\log\frac{dQ}{dP}&=\log\frac{q_0(\xx(\tstart))}{p_0(\xx(\tstart))}+\int_{\tstart}^{\tend}\left(\beta_+^Q-\beta_+^P\right) d\xx+\int_{\tstart}^{\tend}\frac{1}{2}\left(\|\beta_+^P\|^2-\|\beta_+^Q\|^2\right)dt,\quad Q \; {\rm a.s.}, \\
&=\log\frac{q_0(\xx(\tstart))}{p_0(\xx(\tstart))}+\int_{\tstart}^{\tend}\left(\beta_+^Q-\beta_+^P\right) dW^+_t+\int_{\tstart}^{\tend}\frac{1}{2}\|\beta_+^P-\beta_+^Q\|^2dt,\quad Q \; {\rm a.s.}.
\end{split}
\end{equation}
Taking $E_Q$ on both sides, one gets (\ref{RE1}) provided the stochastic integral has zero expectation. In general,
$$\int_{\tstart}^{\tend}\left(\beta_+^Q-\beta_+^P\right) dw_+(t)
$$
is only a local,  ${\cal F}_t^+$ -  martingale. In order to claim that it has zero expectation one needs to ``localize" \cite[p.36]{KS}. Similarly, one can show \eqref{RE2}.
\subsection{The Schr\"{o}dinger bridge problem}
Now let $\rho_0$ and
$\rho_1$ be two
everywhere positive probability densities. Let $\D(\rho_0,\rho_1)$
denote the set of distributions in $\D$ having the prescribed marginal
densities at $\tstart$
and $\tend$. Given $P\in\D$, we consider the following problem:

\begin{equation}\label{problem}{\rm Minimize}\quad H(Q,P) \quad {\rm over} \quad \D(\rho_0,\rho_1).
\end{equation}
 If there is at least one $Q$ in
$\D(\rho_0,\rho_1)$ such that
$H(Q,P)<\infty$,
there exists a unique minimizer $Q^*$ in
$\D(\rho_0,\rho_1)$ called
{\em the Schr\"{o}dinger bridge} from $\rho_0$ to $\rho_1$ over $P$. Indeed, let
$$P_x^y=P\left[\,\cdot\mid \xx(\tstart)=x,\xx(\tend)=y\right],\quad Q_x^y=Q\left[\,\cdot\mid\xx(\tstart)=x,\xx(\tend)=y\right]
$$
be the disintegrations of P and Q with respect to the initial and final positions. Let also
$$\mu^P=P\left[(\xx(\tstart),\xx(\tend))\in(\cdot)\right],\quad \mu^Q=Q\left[(\xx(\tstart),\xx(\tend))\in(\cdot)\right]
$$
be the joint initial-final time distributions under $P$ and $Q$, respectively. Then, we have
$$P=\int P_x^y(\cdot)\mu^P(dx,dy),\quad Q=\int Q_x^y(\cdot)\mu^Q(dx,dy).
$$
By the multiplication formula,
$$\frac{dQ}{dP}=\frac{d\mu^Q}{d\mu^P}(\xx(\tstart),\xx(\tend))\frac{dQ^{\xx(\tend)}_{\xx(\tstart)}}{dP^{\xx(\tend)}_{\xx(\tstart)}},\quad Q \; {\rm a.s.},
$$
we get
\begin{equation}
\begin{split}
H(Q,P)=E_Q\left[\log\frac{dQ}{dP}\right]=E_Q\left[\log\frac{d\mu^Q}{d\mu^P}(\xx(\tstart),\xx(\tend))\right]+E_Q\left[\log\frac{dQ^{\xx(\tend)}_{\xx(\tstart)}}{dP^{\xx(\tend)}_{\xx(\tstart)}}(x)\right]=\\ \int\left(\log\frac{d\mu^Q}{d\mu^P}\right)d\mu^Q+\int\int\left(\log\frac{dQ^{y}_{x}}{dP^{y}_{x}}\right)dQ^y_x \mu^Q(dx,dy).
\end{split}
\end{equation}
This is the sum of two nonnegative quantities. The second becomes zero if and only if
$$Q_x^y=P_x^y, \quad \mu^Q \;{\rm a.s.}.
$$
Thus, as already observed by Schr\"{o}dinger, the problem reduces to minimizing
\begin{equation}\int\left(\log\frac{d\mu^Q}{d\mu^P}\right)d\mu^Q
\end{equation}
subject to the (linear) constraints
\begin{equation}\mu^Q(dx\times\R^n)=\rho_0(x)dx,\quad \mu^Q(\R^n\times dy)=\rho_1(y)dy.
\end{equation}
By the results of Beurlin-Jamison-F\"{o}llmer, this problem has a unique solution $\mu^*$ and
$$Q^*=\int P_x^y(\cdot)\mu^*(dx,dy)
$$
solves (\ref{problem}).
\section{A stochastic control formulation}\label{SC}
Consider now the case where (the coordinate
process under) $P$ is a Markovian diffusion with forward drift field $b_+^P(x,t)$ and
transition density
$p(\sigma,x,\tau,y)$.  The one-time density $\rho(x,t)$ of $P$ is a weak solution of the Fokker-Planck equation
\begin{equation}\label{FP}
\frac{\partial\rho}{\partial t}+\nabla\cdot\left(b_+^P\rho\right)-\frac{1}{2}\Delta\rho=0.
\end{equation}
Moreover, forward and backward drifts are related through Nelson's relation \cite{N1}
\begin{equation}\label{duality}b_-^P(x,t)=b_+^P(x,t)-\nabla\log\rho(x,t).
\end{equation}
Then $Q^*$ is also Markovian with forward drift field
\begin{equation}\label{optdrift}b_+^{Q^*}(x,t)=b_+^P(x,t)+\nabla\log\varphi(x,t),
\end{equation}
 where the
(everywhere positive) function
$\varphi$ solves together
with another function $\hat{\varphi}$ the system
\begin{eqnarray}\label{SY1}
&&\varphi(t,x)=\int
p(t,x,\tend,y)\varphi(\tend,y)dy,\\&&\hat{\varphi}(t,x)=\int
p(\tstart,y,t,x)\hat{\varphi}(\tstart,y)dy.\label{SY2}
\end{eqnarray}
 with boundary conditions
$$\varphi(x,\tstart)\hat{\varphi}(x,\tstart)=\rho_0(x),\quad
\varphi(x,\tend)\hat{\varphi}(x,\tend)=\rho_1(x). $$
Moreover,  the one-time density $\tilde{\rho}$ of $Q^*$ satisfies  the factorization
\begin{equation}\label{factor}\tilde{\rho}(x,t)=\varphi(x,t)\hat{\varphi}(x,t), \forall t\in [\tstart,\tend].
\end{equation} Let us give an elementary derivation of (\ref{optdrift}). Let $\varphi(x,t)$ be any positive, space-time harmonic function, namely $\varphi$ satisfies on $\R^n\times [\tstart,\tend]$
\begin{equation}\label{sth}\frac{\partial\varphi}{\partial t}+b_+^P\cdot\nabla\varphi+\frac{1}{2}\Delta\varphi=0.
\end{equation}
It follows that $\log\varphi$ satisfies
\begin{equation}\label{logeq}
\frac{\partial\log\varphi}{\partial t}+b_+^P\cdot\nabla\log\varphi+\frac{1}{2}\Delta\log\varphi=-\frac{1}{2}\|\nabla\log\varphi\|^2.
\end{equation}
Observe now that, in view of (\ref{RE1}), problem (\ref{problem}) is equivalent to minimizing over $\D(\rho_0,\rho_1)$ the functional
\begin{equation}\label{funct}I(Q)=E_Q\left[\int_{\tstart}^{\tend}\frac{1}{2}
\|\beta_+^Q-b_+^P(\xx(t),t)\|^2dt -\log\varphi(\xx(\tend),\tend)+\log\varphi(\xx(\tstart),\tstart)\right].
\end{equation}
This follows from the fact that $H(Q,P)$ and (\ref{funct}) differ by a quantity which is constant over $\D(\rho_0,\rho_1)$. Observe now that, under $Q$, by Ito's rule,
\begin{equation}d\log\varphi(\xx(t),t)=\left[\frac{\partial\log\varphi}{\partial t}(\xx(t),t)+\beta_+^Q\cdot\nabla\log\varphi(\xx(t),t)+\frac{1}{2}\Delta\log\varphi(\xx(t),t)\right]dt+\nabla\log\varphi(\xx(t),t)dw_+(t).
\end{equation}
Using this and (\ref{logeq}) in (\ref{funct}), we now get
\begin{eqnarray}\nonumber
I(Q)&=& E_Q\left[\int_{\tstart}^{\tend}\frac{1}{2}
\|\beta_+^Q-b_+^P(\xx(t),t)\|^2dt -\log\varphi(\xx(\tend),\tend)+\log\varphi(\xx(\tstart),\tstart)\right]\\\nonumber
&=& E_Q\left[\int_{\tstart}^{\tend}\left(\frac{1}{2}
\|\beta_+^Q-b_+^P(x(t),t)\|^2-\left[\frac{\partial\log\varphi}{\partial t}+\beta_+^Q\cdot\nabla\log\varphi+\frac{1}{2}\Delta\log\varphi\right](\xx(t),t) \right)dt\right.\\\nonumber
&&\left.-\int_{\tstart}^{\tend}\nabla\log\varphi(\xx(t),t)dw_+(t)\right]\\\nonumber
&=&E_Q\left[\int_{\tstart}^{\tend}\left(\frac{1}{2}
\|\beta_+^Q-b_+^P(x(t),t)\|^2-\left(\beta_+^Q-b_+^P(x(t),t)\right)\cdot\nabla\log\varphi(\xx(t),t)+\frac{1}{2}\|\nabla\log\varphi(\xx(t),t)\|^2 \right)dt\right]\\
&=&E_Q\left[\int_{\tstart}^{\tend}\frac{1}{2}
\|\beta_+^Q-b_+^P(\xx(t),t)-\nabla\log\varphi(\xx(t),t)\|^2dt\right],
\end{eqnarray}
where again we assumed that the stochastic integral has zero expectation. Then the form (\ref{optdrift}) of the forward drift of $Q^*$ follows. Define now
$$\hat{\varphi}(x,t)=\frac{\tilde{\rho}(x,t)}{\varphi(x,t)}.
$$
Then a direct calculation using (\ref{sth}), and the Fokker-Planck equation satisfied by $\tilde{\rho}$
\begin{equation}\label{FP2}
\frac{\partial\tilde{\rho}}{\partial t}+\nabla\cdot\left((b_+^P+\nabla\log\varphi)\tilde{\rho}\right)-\frac{1}{2}\Delta\tilde{\rho}=0,
\end{equation}
yields
\begin{equation}\label{coharmonic}
\frac{\partial\hat{\varphi}}{\partial t}+ \nabla\cdot\left(b_+^P\hat{\varphi}\right)-\frac{1}{2}\Delta\hat{\varphi}=0.
\end{equation}
Thus, $\hat{\varphi}$ is co-harmonic, namely it satisfies the original Fokker-Planck equation (\ref{FP}) just like $\rho(x,t)$, the one-time density of the ``prior" $P$.

Suppose we start instead with $\psi(x,t)$, a positive, reverse-time space-time harmonic function, namely $\psi$ satisfies on $\R^n\times [\tstart,\tend]$
\begin{equation}\label{rtsth}\frac{\partial\psi}{\partial t}+b_-^P\cdot\nabla\psi-\frac{1}{2}\Delta\psi=0,
\end{equation}
where $b_-^P(x)=b_+^P(x)-\nabla\log\rho(x,t)$ is the backward drift of $P$. Then $\log\psi$ satisfies
\begin{equation}\label{backlogeq}
\frac{\partial\log\psi}{\partial t}+b_-^P\cdot\nabla\log\psi-\frac{1}{2}\Delta\log\psi=\frac{1}{2}\|\nabla\log\psi\|^2.
\end{equation}
Consider now the functional
\begin{equation}\label{backfunct}\bar{I}(Q)=E_Q\left[\int_{\tstart}^{\tend}\frac{1}{2}
\|\beta_-^Q-b_-^P(\xx(t),t)\|^2dt +\log\psi(\xx({\tend}),\tend)-\log\psi(\xx({\tstart}),\tstart)\right].
\end{equation}
Again, minimizing $\bar{I}(Q)$ over $\D(\rho_0,\rho_1)$ is equivalent to  (\ref{problem}). By Ito's rule, under $Q$, we have
\begin{equation}d\log\psi(\xx(t),t)=\left[\frac{\partial\log\psi}{\partial t}+\beta_-^Q\cdot\nabla\log\psi+\frac{1}{2}\Delta\log\psi\right](\xx(t),t)dt+\nabla\log\psi(\xx(t),t)dw_-(t).
\end{equation}
Using this differential in (\ref{backfunct}), we now get
\begin{eqnarray*}
\bar{I}(Q)&=& E_Q\left[\int_{\tstart}^{\tend}\frac{1}{2}
\|\beta_-^Q-b_-^P(\xx(t),t)\|^2dt +\log\psi(x(\tend),\tend)-\log\psi(x(\tstart),\tstart)\right] \\
&=& E_Q\left[\int_{\tstart}^{\tend}\left(\frac{1}{2}
\|\beta_-^Q-b_-^P(\xx(t),t)\|^2+\left[\frac{\partial\log\psi}{\partial t}+\beta_-^Q\cdot\nabla\log\psi-\frac{1}{2}\Delta\log\psi\right](\xx(t),t) \right)dt\right.\\
&&\left.+\int_{\tstart}^{\tend}\nabla\log\psi(\xx(t),t)dw_-(t)\right]\\
&=&E_Q\left[\int_{\tstart}^{\tend}\left(\frac{1}{2}
\|\beta_-^Q-b_-^P(\xx(t),t)\|^2+\left(\beta_-^Q-b_-^P(\xx(t),t)\right)\cdot\nabla\log\psi(\xx(t),t)+\frac{1}{2}\|\nabla\log\psi(\xx(t),t)\|^2 \right)dt\right]\\
&=&E_Q\left[\int_{\tstart}^{\tend}\frac{1}{2}
\|\beta_-^Q-b_-^P(\xx(t),t)+\nabla\log\psi(\xx(t),t)\|^2dt\right].
\end{eqnarray*}
We then get
\begin{equation}\label{rtopt}
b_-^{Q^*}(x,t)=b_-^P(x,t)-\nabla\log\psi(x,t).
\end{equation}
Thus, the solution $Q^*$ is, in the language of Doob, an $h$-path process both in the forward and in the backward direction of time. We now identify $\psi$. By (\ref{factor}), we have
%{\color{red} (why?)}
\begin{eqnarray}\nonumber
\tilde{\rho}(x,t)&=&\varphi(x,t)\hat{\varphi}(x,t)=\varphi(x,t)\frac{\hat{\varphi}(x,t)}{\rho(x,t)}\rho(x,t)\\&=&\varphi(x,t)\psi(x,t)\rho(x,t),\quad\psi(x,t)=\frac{\hat{\varphi}(x,t)}{\rho(x,t)},\quad \forall t\in [\tstart,\tend].
\label{newfactor}
\end{eqnarray}
Indeed, $\psi$, being the ratio of two solutions of the original Fokker-Planck (\ref{coharmonic}), is reverse-time space-time harmonic, it namely satisfies (\ref{rtsth}) \cite{P89}. This agrees with the following calculation using (\ref{rtopt}), (\ref{duality}), (\ref{newfactor}) and (\ref{optdrift})
\begin{equation}
\begin{split}b_-^{Q^*}(x,t)&=b_-^P(x,t)-\nabla\log\psi(x,t)\\
                                        &=b_+^P(x,t)-\nabla\log\rho(x,t)-\nabla\log\psi(x,t)\\
                                        &=b_+^P(x,t)-\nabla\log\hat{\varphi}(x,t)=b_+^P(x,t)\pm\nabla\log\varphi(x,t)-\nabla\log\hat{\varphi}(x,t)\\
                                        &=b_+^{Q^*}(x,t)-\nabla\log\tilde{\rho}(x,t),
\end{split}
\end{equation}
which is simply Nelson's duality relation for the drifts of $Q^*$. Formula (\ref{newfactor}) should be compared to \cite[Theorem 3.4]{leo}.

Finally, there are also conditional versions of these variational problems which are closer to standard stochastic control problems. Consider, for instance,
minimizing the functional
\begin{equation}\label{controlfunct}J(u)=E_{tx}\left[\int_{t}^{\tend}\frac{1}{2}
\|u(t)\|^2dt -\log\varphi_1(\xx(\tend))\right],
\end{equation}
$$dx(t)=\left[b_+^P(\xx(t),t)+u(x(t),t)\right]dt+dw_+(t), \xx(t)=x \; {\rm a.s.}.
$$
over feedback controls $u$ such that the differential equation has a weak solution. If $\varphi(x,t)$ solves (\ref{sth}) with terminal condition $\varphi_1(x)$, then, the same argument as before shows that $u^*(x,t)=\nabla\log\varphi(x,t)$ is optimal and that $S(x,t)=-\log\varphi(x,t)={\rm inf}_u J(u)$ is the {\em value function} of the control problem. By (\ref{logeq}), the Hamilton-Jacobi-Bellman equation has the form
$$\frac{\partial S}{\partial t}+{\rm inf}_u\left[\left(b_+^P+u\right)\cdot\nabla S+\frac{1}{2}\|u\|^2\right]+\frac{1}{2}\Delta S=0, \quad S(x,\tend)=-\log\varphi_1(x).$$

\section{A time-symmetric formulation}\label{TSF}
Inspired by a paper by Nagasawa \cite{Nag}, we proceed to derive a control time-symmetric formulation of the bridge problem. For any $Q\in\D$, define the {\em current} and {\em osmotic} drifts
$$v^Q=\frac{\beta_+^Q+\beta_-^Q}{2},\quad u^Q=\frac{\beta_+^Q-\beta_-^Q}{2}.
$$
Then
$$\beta_+^Q=v^Q+u^Q,\quad \beta_-^Q=v^Q-u^Q.
$$
Observe that
\begin{eqnarray}
\nonumber H(Q,P)&=&\frac{1}{2}H(q(\tstart),p(\tstart))+\frac{1}{2}H(q(\tend),p(\tend))\\\nonumber&&+E_Q\left[\int_{\tstart}^{\tend}\frac{1}{4}
\|\beta_+^Q-\beta_+^P\|^2+\frac{1}{4}
\|\beta_-^Q-\beta_-^P\|^2dt\right]\\\label{symm}&=&\frac{1}{2}H(q(\tstart),p(\tstart))+\frac{1}{2}H(q(\tend),p(\tend))\nonumber\\&&+E_Q\left[\int_{\tstart}^{\tend}\frac{1}{2}
\|v^Q-v^P\|^2+\frac{1}{2}
\|u^Q-u^P\|^2dt\right].
\end{eqnarray}
Since {$H(q_0,p_0)$ and $H(q_1,p_1)$} are constant over $\D(\rho_0,\rho_1)$, it follows that the Schr\"{o}dinger bridge  $Q^*$ minimizes the sum of the two incremental kinetic energies.
Finally, we consider minimizing over $\D(\rho_0,\rho_1)$ the functional
$$I_s(Q)=\frac{1}{2}\left[I(Q)+\bar{I}(Q)\right].
$$
By the previous calculation, this is equivalent to  minimizing over $\D(\rho_0,\rho_1)$ the functional
\begin{equation}\label{symmfunct} E_Q\left[\int_{\tstart}^{\tend}\left(\frac{1}{2}
\|v^Q-v^P\|^2+\frac{1}{2}\|u^Q-u^P\|^2\right)dt -\frac{1}{2}\log\frac{\varphi}{\psi}(x(\tend),\tend)+\frac{1}{2}\log\frac{\varphi}{\psi}(x(\tstart),\tstart)\right].
\end{equation}
The following current and osmotic drifts make the functional equal to zero and are therefore optimal
\begin{eqnarray}\label{optv}v^{Q^*}(x,t)=v^P(x,t)+\frac{1}{2}\nabla\log\frac{\varphi}{\psi}(x,t),\\u^{Q^*}(x,t)=u^P(x,t)+\frac{1}{2}\nabla\log (\varphi\psi)(x,t),
\label{optu}
\end{eqnarray}
which agree with (\ref{optdrift}) and (\ref{rtopt}). A variational analysis with the two controls $v$ and $u$ can be developed along the lines of \cite[Sections III-IV]{P}.

\section{A fluid dynamic formulation of the Schr\"{o}dinger bridge problem}\label{FDFS}

Let us go back to the symmetric representation (\ref{symm}). In the case where the prior measure is $P=W$ stationary Wiener measure, we have $v^W=u^W=0$ \footnote{See \cite[pp. 7-8]{leo} for a justification of employing unbounded path measures in relative entropy problems.}. It basically corresponds to the situation where there is no prior information. Considering that the boundary relative entropies are constant, we get that the problem is equivalent to minimizing
$$\E\left\{\int_{\tstart}^{\tend}\left[\frac{1}{2}\|v\|^2+\frac{1}{2}\|u\|^2\right]dt\right\}
$$
over $\D(\rho_0,\rho_1)$. Let us restrict our search to Markovian processes and recall Nelson's duality formula relating the two drifts
\begin{equation}\label{Nelsonduality}u(x,t)=\frac{1}{2}\nabla\log\rho(x,t).
\end{equation}
where $u$ is the osmotic drift field, and the current drift field
$$v(x,t)=\frac{b_+(x,t)+b_-(x,t)}{2}.
$$
Then,
\begin{equation}\label{conteq}
\frac{\partial \rho}{\partial t}+\nabla\cdot(v\rho)=0.
\end{equation}
Thus, we get that the problem is equivalent to minimizing
\begin{eqnarray}\label{SBB1}&&\inf_{(\rho,v)}\int_{\R^n}\int_{\tstart}^{\tend}\left[\frac{1}{2}\|v(x,t)\|^2+\frac{1}{8}\|\nabla\log\rho(x,t)\|^2\right]\rho(t,x)dtdx,\\&&\frac{\partial \rho}{\partial t}+\nabla\cdot(v\rho)=0,\label{SBB2}\\&& \rho(\tstart,x)=\rho_0(x), \quad \rho(\tend,y)=\rho_1(y),\label{Sboundary}
\end{eqnarray}
which should be compared to (\ref{BB1})-(\ref{BB2})-(\ref{boundary}). We notice, in particular, that the two functionals differ by a term which is a multiple of the integral over time of the Fisher information functional
$$\int_{\R^n}\|\nabla\log\rho(x,t)\|^2\rho(t,x)dx.
$$
This unveils a relation between the two problems without zero noise limits \cite{Mik, leo2}.

Finally, we mention that a fluid dynamic problem concerning swarms of particles diffusing anisotropically with losses has been proposed and studied in \cite{CGP2}. It may or may not have a probabilistic counterpart as a Schr\"{o}dinger bridge problem.

\section{Optimal transport with a ``prior"}\label{OMTprior}

Considering the relation we have seen between the fluid dynamic versions of the optimal transport problem and the Schr\"{o}dinger bridge problem, one may wonder whether there exists a formulation of the former which allows for an  ``a priori" evolution like in the latter.  Relative entropy on path space does not work for zero-noise random evolutions as they are singular. Indeed, let $P_\epsilon$ and $Q_\epsilon$ be the measures on { $C([\tstart,\tend],\R^n)$} equivalent to stationary Wiener measure W with forward differentials
\begin{eqnarray}\nonumber
d\xx(t)&=&\beta_+^{P_\epsilon} dt+\sqrt{\epsilon}dw_+(t),\\
d\xx(t)&=&\beta_+^{Q_\epsilon} dt+\sqrt{\epsilon}dw_+(t).
\end{eqnarray}
Then, one can argue along the same lines as in Section \ref{background} that
{
$$H(Q_\epsilon,P_\epsilon)=H(q_0,p_0)+\E_{Q_\epsilon}\left[\int_{\tstart}^{\tend}\frac{1}{2\epsilon}\|\beta_+^{Q_\epsilon}-\beta_+^{P_\epsilon}\|^2dt\right].
$$}
For $\epsilon\searrow 0$, the relative entropy becomes infinite unless $Q_\epsilon=P_\epsilon$ \footnote{{ This calculation indicates that there may be a limit as $\epsilon\searrow 0$  of  ${\rm inf}\{\epsilon H(Q_\epsilon,P_\epsilon)\}$ and, hopefully, in suitable sense, of the minimizers. This is indeed the case, see \cite{Mik,leo2,leo} for a precise statement of limiting results.}}. We need therefore to take a different route, namely start with the following fluid dynamic control problem. Suppose we have two probability densities $\rho_0$ and $\rho_1$ and a flow of probability densities $\{\rho(x,t); \tstart\le t\le \tend\}$ satisfying
\begin{equation}\label{continuityeq}
\frac{\partial \rho}{\partial t}+\nabla\cdot(v\rho)=0,
\end{equation}
for some continuos vector field $v(\cdot,\cdot)$. We take (\ref{continuityeq}) as our ``prior" evolution and formulate the following problem. Let
\begin{subequations}\label{eq:priorBB}
\begin{eqnarray}\label{priorBB1}&&\inf_{(\tilde{\rho},\tilde{v})}\int_{\R^n}\int_{\tstart}^{\tend}\frac{1}{2}\|\tilde{v}(x,t)-v(x,t)\|^2\tilde{\rho}(t,x)dtdx,\\&&\frac{\partial \tilde{\rho}}{\partial t}+\nabla\cdot(\tilde{v}\tilde{\rho})=0,\label{priorBB2}\\&&\tilde{\rho}(\tstart,x)=\rho_0(x), \quad \tilde{\rho}(\tend,y)=\rho_1(y).\label{priorboundary}
\end{eqnarray}
\end{subequations}
Clearly, if the prior flow satisfies $\rho(x,\tstart)=\rho_0(x)$ and $\rho(x,\tend)=\rho_1(x)$, then it solves the problem and $\tilde{v}^*=v$. Moreover, the standard optimal transport problem is recovered when $v\equiv 0$, namely the prior evolution is constant in time.

Let us try to provide further motivation to study problem \eqref{eq:priorBB}. Consider the situation where a previous optimal transport problem \eqref{eq:BB} has been solved with boundary marginals $\bar{\rho}_0$ and $\bar{\rho}_1$ leading to the optimal velocity field $v(x,t)$. Here say $\bar{\rho}_0$ represent resources being produced to satisfy the demand $\bar{\rho}_1$. Suppose now new information becomes available showing that the actual resources available are distributed according to $\rho_0$ and the actual demand is distributed according to $\rho_1$. As we had already set up a transportation plan according to velocity field $v$, we seek to solve a new transport problem where the new evolution is close to the one we would have employing the previous velocity field. This is represented in problem \eqref{eq:priorBB}.
{%\mike
\begin{remark}
The particle version of \eqref{eq:priorBB} takes the form of a more familiar OMT problem, namely, in the notation of Section \ref{BB},
\begin{equation}\label{GenOptTrans}
\inf_{\pi\in\Pi(\mu,\nu)}\int_{\R^n\times\R^n}c(x,y)d\pi(x,y),
\end{equation}
where
\begin{equation}\label{cost}
c(x,y)=\inf_{\x\in\mathcal X_{xy}}\int_{\tstart}^{\tend}L(t,x(t),\dot{x}(t))dt,\quad L(t,x,\dot{x})=\|\dot{x}-v(x,t)\|^2.
\end{equation}
The explicit calculation of the function $c(x,y)$ when $v\not\equiv 0$ is nontrivial. Moreover, the zero noise limit results of \cite[Section 3]{leo2}, based on a Large Deviations Principle \cite{DZ}, although very general in other ways, seem to cover here only the case where $c(x,y)=c(x-y)$ strictly convex originating from a Lagrangian $L(t,x,\dot{x})=c(\dot{x})$. Finally, we feel that our formulation is a most natural one in which to study zero noise limits of Schroedinger bridges with a general Markovian prior evolution. In the next section, we discuss this problem in the Gaussian case. The proof of the convergence of the path-space measures of the minimisers can be done along the lines of \cite{leo2} where $\Gamma$-convergence of the bridge minimum problems to the OMT problem is established. This, under suitable assumptions, guarantees convergence of the minimizers.
\end{remark}}

The variational analysis for \eqref{eq:priorBB} can be carried out as in Section \ref{BB}. Let $\mathcal P_{\rho_0\rho_1}$ be again the family of flows of probability densities $\rho=\{\rho(\cdot,t); \tstart\le t\le \tend\}$ satisfying (\ref{priorboundary}). Let $\mathcal V$ be the family of continuous feedback control laws $\tilde{v}(\cdot,\cdot)$. Consider the unconstrained minimization of the Lagrangian over $\mathcal P_{\rho_0\rho_1}\times\mathcal V$
\begin{equation}\label{priorlagrangian}
\mathcal L(\tilde{\rho},v)=\int_{\R^n}\int_{\tstart}^{\tend}\left[\frac{1}{2}\|\tilde{v}(x,t)-v(x,t)\|^2\tilde{\rho}(t,x)+\lambda(x,t)\left(\frac{\partial \tilde{\rho}}{\partial t}+\nabla\cdot(\tilde{v}\tilde{\rho})\right)\right]dtdx,
\end{equation}
where again $\lambda$ is a $C^1$ Lagrange multiplier. After integration by parts, assuming that limits for $x\rightarrow\infty$ are zero, and observing that the boundary values are constant over $\mathcal P_{\rho_0\rho_1}$, we get the problem
\begin{equation}\label{priorlagrangian2}\inf_{(\tilde{\rho},\tilde{v})\in\mathcal P_{\rho_0\rho_1}\times\mathcal V}\int_{\R^n}\int_{\tstart}^{\tend}\left[\frac{1}{2}\|\tilde{v}(x,t)-v(x,t)\|^2+\left(-\frac{\partial \lambda}{\partial t}-\nabla\lambda\cdot \tilde{v}\right)\right]\tilde{\rho}(x,t)dtdx
\end{equation}
Pointwise minimization with respect to $\tilde{v}$ for each fixed flow of probability densities $\tilde{\rho}$ gives
\begin{equation}\label{prioroptcond}
v^*_\rho(x,t)=v(x,t)+\nabla\lambda(x,t).
\end{equation}
Plugging this form of the optimal control into (\ref{priorlagrangian2}), we get the functional of $\tilde{\rho}\in\mathcal P_{\rho_0\rho_1}$
\begin{equation}\label{priorlagrangian3}
J(\tilde{\rho})=-\int_{\R^n}\int_{\tstart}^{\tend}\left[\frac{\partial \lambda}{\partial t}+v\cdot\nabla\lambda+\frac{1}{2}\|\nabla\lambda\|^2\right]\tilde{\rho}(x,t)dtdx.
\end{equation}
We then have the following result:
\begin{proposition}If $\tilde{\rho}^*$ satisfying
\begin{equation}\label{prioroptev}
\frac{\partial \tilde{\rho}^*}{\partial t}+\nabla\cdot[(v+\nabla\psi)\tilde{\rho}^*]=0, \quad \tilde{\rho}^*(x,\tstart)=\rho_0(x),
\end{equation}
where $\psi$ is a solution of the Hamilton-Jacobi equation
\begin{equation}\label{eq:hamiltonjacobi}
\frac{\partial \psi}{\partial t}+v\cdot\nabla\psi+\frac{1}{2}\|\nabla\psi\|^2=0,
\end{equation}
is such that $\tilde{\rho}^*(x,\tend)=\rho_1(x)$, then the pair $\left(\tilde{\rho}^*(x,t),v^*(x,t)=v(x,t)+\nabla\psi(x,t)\right)$ is a solution of the problem \eqref{eq:priorBB}.
\end{proposition}
If $v(x,t)=\alpha(t)x$, and both $\rho_0$ and $\rho_1$ are Gaussian, then the optimal evolution is given by a linear equation and is therefore given by a Gaussian process as we will study next.

\section{Gaussian case}\label{gaussian}
In this section, we consider the correspondence between Schr\"{o}dinger bridges and optimal mass transport for the special case where the underlying dynamics are linear and the marginals are normal distributions. To this end, consider { the reference evolution}
	\begin{equation}\label{eq:dynamicslinear}
		d\xx(t)=A(t)\xx(t)dt+\sqrt{\epsilon}dw(t)
	\end{equation}
and the two marginals
	\begin{subequations}\label{eq:marginalslinear}
	\begin{equation}
		\rho_0(x)=(2\pi|\Sigma_0|)^{-n/2}\exp\left[-\frac{1}{2}(x-m_0)'\Sigma_0^{-1}(x-m_0)\right],
	\end{equation} 		
	\begin{equation}
		\rho_1(x)=(2\pi|\Sigma_1|)^{-n/2}\exp\left[-\frac{1}{2}(x-m_1)'\Sigma_1^{-1}(x-m_1)\right],
	\end{equation}
	\end{subequations}
{ where prime denotes transposition.} In our previous work \cite{CGP}, we derived a ``closed form'' expression for the corresponding Schr\"{o}dinger bridge for the case when $m_0=m_1=0$, namely,
	\begin{equation}\label{eq:schrodingerbridgelinear1}
		d\xx(t)=(A(t)-\Pi_\epsilon(t))\xx(t) dt+\sqrt{\epsilon}dw(t)
	\end{equation}
with $\Pi_\epsilon(t)$ satisfying the matrix Riccati equation
	\begin{equation}\label{eq:schrodingerbridgefeedback}
		\dot{\Pi}_\epsilon(t)+A(t)'\Pi_\epsilon(t)+\Pi_\epsilon(t)A(t)-\Pi_\epsilon(t)^2=0
	\end{equation}
and the boundary condition
	\[
        \Pi_\epsilon(0)=\Sigma_0^{-1/2}[\frac{\epsilon}{2}I+
	\Sigma_0^{1/2}\Phi_{10}^\prime M_{10}^{-1}\Phi_{10}\Sigma_0^{1/2}-
	(\frac{\epsilon^2}{4}I+\Sigma_0^{1/2}\Phi_{10}^\prime M_{10}^{-1}\Sigma_1M_{10}^{-1}
	\Phi_{10}\Sigma_0^{1/2}
	)^{1/2}]\Sigma_0^{-1/2}.
        \]
Here $\Phi_{10}:=\Phi(\tend,\tstart)$ is the state transition matrix from $\tstart$ to $\tend$ and
	\[
		M_{10}:=M(\tend,\tstart)=\int_{\tstart}^{\tend}\Phi(\tend,t)\Phi(\tend,t)'dt
	\]
is the controllability gramian. This can be easily adjusted for the case when $m_0\neq 0$ or  $m_1\neq 0$ by adding an extra deterministic drift term to account for the change in the mean as follows:
	\begin{equation}\label{eq:schrodingerbridgelinear2}
		d\xx(t)=(A(t)-\Pi_\epsilon(t))\xx(t) dt+m(t)dt+\sqrt{\epsilon}dw(t)
	\end{equation}
where
	\begin{equation}\label{eq:schrodingerbridgedrift}
		m(t)=\hat{\Phi}(\tend,t)'\hat{M}(\tend,\tstart)^{-1}(m_1-\hat{\Phi}(\tend,\tstart)m_0)
	\end{equation}
with $\hat{\Phi}(t,s), \hat{M}(t,s)$ satisfying
	\begin{equation*}
		\frac{\partial \hat{\Phi}(t,s)}{\partial t} =(A(t)-\Pi_\epsilon(t))\hat{\Phi}(t,s),~~~\hat{\Phi}(t,t)=I
	\end{equation*}
and
	\[
		\hat{M}(t,s)=\int_s^t \hat{\Phi}(t,\tau)\hat{\Phi}(t,\tau)'d\tau.
	\]
We now consider ``slowing down'' the reference evolution
%stochastic drift
by letting $\epsilon$ go to $0$. In the case where $A(t)\equiv 0$, the { Schr\"{o}dinger bridge solution }process converges to the solution of optimal mass transport problem ~\cite{Mik,leo}. In general, { when $A(t)\not\equiv 0$,} by taking $\epsilon= 0$ we obtain
	\begin{equation}\label{eq:schrodingerbridgeinitial}
		\Pi_0(0)=\Sigma_0^{-1/2}[\Sigma_0^{1/2}\Phi_{10}^\prime M_{10}								^{-1}\Phi_{10}\Sigma_0^{1/2}-(\Sigma_0^{1/2}\Phi_{10}^\prime M_{10}		^{-1}\Sigma_1M_{10}^{-1}\Phi_{10}\Sigma_0^{1/2})^{1/2}]\Sigma_0^{-1/2},
	\end{equation}
and a limiting process
	\begin{equation}\label{eq:optimaltransportprior}
		dx(t)=(A(t)-\Pi_0(t))x(t) dt+m(t)dt,~~x(\tstart)\sim (m_0,\Sigma_0)
	\end{equation}
with $\Pi_0(t), m(t)$ satisfying \eqref{eq:schrodingerbridgefeedback}, \eqref{eq:schrodingerbridgedrift} and \eqref{eq:schrodingerbridgeinitial}. In fact $\Pi_0(t)$ has explicit expression
    \begin{eqnarray}
        \nonumber
        \Pi_0(t)&=&-M(t,\tstart)^{-1}-M(t,\tstart)^{-1}\Phi(t,\tstart)\left[\Phi_{10}^\prime M_{10}^{-1}\Phi_{10}\right.
        \\
        &&
        \left.-\Sigma_0^{-1/2}(\Sigma_0^{1/2}\Phi_{10}^\prime M_{10}^{-1}\Sigma_1M_{10}^{-1}\Phi_{10}\Sigma_0^{1/2})^{1/2}\Sigma_0^{-1/2}-
        \Sigma_0^{-1/2}\right]^{-1}\Phi(t,\tstart)'M(t,\tstart)^{-1}
        \label{eq:feedbackPi}
    \end{eqnarray}
%{\mike Convergence of the path-space law $P_\epsilon$ corresponding to the solution of (\ref{eq:schrodingerbridgelinear2}) to the law of  (\ref{eq:optimaltransportprior}) is guaranteed by the Friedlin-Wentzell results \cite{FW}. Indeed, the $\{P_\epsilon\}$ satisfy a large deviation principle from which convergence of the solutions follows. }
It turns out that process \eqref{eq:optimaltransportprior} yields an optimal solution to the transport problem \eqref{eq:priorBB} as stated next.
\begin{theorem}
Let $\tilde{\rho}(\cdot,t)$ be the probability density of $x(t)$ in \eqref{eq:optimaltransportprior}, and $\tilde{v}(x,t)=(A(t)-\Pi_0(t))x+m(t)$. Then the pair $(\tilde{\rho},\tilde{v})$ is a solution of the problem \eqref{eq:priorBB} with prior { velocity field $v(x,t)=A(t)x$.}

\end{theorem}
\begin{proof}
To show that the pair $(\tilde{\rho},\tilde{v})$ is a solution, we need to prove i) $\tilde{\rho}$ satisfies the boundary condition $\tilde{\rho}(x,\tend)=\rho_1(x)$ and ii) $\tilde{v}(x,t)-v(x,t)=\nabla \psi(x,t)$ for some $\psi$ with $\psi$ satisfying the Hamilton-Jacobi equation \eqref{eq:hamiltonjacobi}. Here $v(x,t)=A(t)x$ is the drift of the prior process.

We first show that $\tilde{\rho}$ satisfies the boundary condition $\tilde{\rho}(x,\tend)=\rho_1(x)$. Since the process \eqref{eq:optimaltransportprior} is a linear diffusion with gaussian initial condition, $x(t)$ is a gaussian random vector for all $t\in [\tstart,\tend]$. Let
    \[
        \tilde{\rho}(x,t)=(2\pi|\Sigma(t)|)^{-n/2}\exp\left[-\frac{1}{2}(x-n(t))'\Sigma(t)^{-1}(x-n(t))\right].
    \]
Then obviously the mean value $n(t)$ is
    \[
        n(t)=\hat{\Phi}(t,\tstart)m_0+\int_{\tstart}^{t}\hat{\Phi}(t,\tau)m(\tau)d\tau.
    \]
We claim that the covariance $\Sigma(t)$ has the explicit expression
    \begin{eqnarray}
        \nonumber
        \Sigma(t)&=&M(t,\tstart)\Phi(\tstart,t)'\Sigma_0^{-1/2}\left[-\Sigma_0^{1/2}\Phi_{10}^\prime M_{10}^{-1}\Phi_{10}\Sigma_0^{1/2}+(\Sigma_0^{1/2}\Phi_{10}^\prime M_{10}^{-1}\Sigma_1M_{10}^{-1}\Phi_{10}\Sigma_0^{1/2})^{1/2}\right.\\
        &&\left.+\Sigma_0^{1/2}\Phi(t,\tstart)'
        M(t,\tstart)^{-1}\Phi(t,\tstart)\Sigma_0^{1/2}\right]^2 \Sigma_0^{-1/2}\Phi(\tstart,t)M(t,\tstart),
        \label{eq:statecovariance}
    \end{eqnarray}
for $t\in(\tstart,\tend]$. This expression is consistent with the initial condition $\Sigma_0$ since
    \[
        \lim_{t\searrow \tstart}\Sigma(t)=\Sigma_0.
    \]
To see that $\Sigma(t)$ is the covariance matrix of $x(t)$, we only need to show that $\Sigma(t)$ satisfies the differential equation
    \[
        \dot{\Sigma}(t)=(A(t)-\Pi_0(t))\Sigma(t)+\Sigma(t)(A(t)-\Pi_0(t))'.
    \]
This can be verified directly from the expression \eqref{eq:statecovariance} and \eqref{eq:feedbackPi} after  some straightforward but lengthy computations. Now observing that
    \begin{eqnarray*}
        n(\tend)&=&\hat{\Phi}(\tend,\tstart)m_0+\int_{\tstart}^{\tend}\hat{\Phi}(\tend,\tau)m(\tau)d\tau\\
        &=&\hat{\Phi}(\tend,\tstart)m_0+\int_{\tstart}^{\tend}\hat{\Phi}(\tend,\tau)\hat{\Phi}(\tend,\tau)'
        d\tau\hat{M}(\tend,\tstart)^{-1}(m_1-\hat{\Phi}(\tend,\tstart)m_0)=m_1
    \end{eqnarray*}
and
    \begin{eqnarray*}
        \Sigma(\tend)&=&M(\tend,\tstart)\Phi(\tstart,\tend)'\Sigma_0^{-1/2}\left[(\Sigma_0^{1/2}\Phi_{10}^\prime M_{10}^{-1}\Sigma_1M_{10}^{-1}\Phi_{10}\Sigma_0^{1/2})^{1/2}\right]^2\Sigma_0^{-1/2}\Phi(\tstart,\tend)M(\tend,\tstart)=\Sigma_1,
    \end{eqnarray*}
we conclude that $\tilde{\rho}$ satisfies the boundary condition $\tilde{\rho}(x,\tend)=\rho_1(x)$.

We next show ii). Let
    \[
        \psi(x,t)=-\frac{1}{2}x'\Pi_0(t)x+m(t)'x+c(t)
    \]
with
    \[
        c(t)=-\frac{1}{2}\int_{\tstart}^{t}m(\tau)'m(\tau)d\tau,
    \]
then
    \begin{eqnarray*}
    \frac{\partial \psi}{\partial t}+v\cdot\nabla\psi+\frac{1}{2}\|\nabla\psi\|^2
    &=& -\frac{1}{2}x'\dot{\Pi}_0(t)x+\dot{m}(t)'x+\dot{c}(t)+x'A(t)'(m(t)-\Pi_0(t)x)+\frac{1}{2}\|m(t)'-x'\Pi_0(t)\|^2\\
    &=&\frac{1}{2}x'(A(t)'\Pi_0+\Pi_0 A(t)-\Pi_0(t)^2)x-m(t)'(A(t)-\Pi_0(t))x+\dot{c}(t)\\
    &&+x'A(t)'(m(t)-\Pi_0(t)x)+\frac{1}{2}(m(t)'-x'\Pi_0(t))(m(t)-\Pi_0(t)x)\\
    &=& \dot{c}(t)+\frac{1}{2}m(t)'m(t)=0.
    \end{eqnarray*}
This completes the proof.
\end{proof}

\section{Example: Shifting the mean of normal distributions}\label{shifting}
{For illustration purposes, }we consider the Schr\"{o}dinger bridge problem on the time interval $[0,1]$ and $x\in\R$ when the ``prior" is $\sigma W_t$ and the two marginals are
\begin{equation}\label{marginals}
\rho_0(x)=(2\pi)^{-1/2}\exp\left[-\frac{x^2}{2}\right], \quad \rho_1(x)=(2\pi)^{-1/2}\exp\left[-\frac{(x-1)^2}{2}\right].
\end{equation}
By the general theory, we know that the bridge has forward differential
$$dx(t)=\frac{\partial}{\partial x}\log\varphi(x(t),t)dt+\sigma dw(t)
$$
where $\varphi$ solves together with $\hat{\varphi}$ the Schr\"{o}dinger system
\begin{equation}\label{Schroedingersystem}\left\{\begin{array}{ll}\frac{\partial \varphi}{\partial t}+\frac{\sigma^2}{2}\Delta\varphi=0, \quad \varphi(x,0)\hat{\varphi}(x,0)=\rho_0(x)\\\frac{\partial \hat{\varphi}}{\partial t}-\frac{\sigma^2}{2}\Delta\hat{\varphi}=0,  \quad \varphi(x,1)\hat{\varphi}(x,1)=\rho_1(x).
\end{array}\right.,
\end{equation}
{It can be seen that}
$$dx(t)=\left[\frac{\sigma^2}{\sigma^2t+c}x(t)+\frac{c}{\sigma^2t+c}\right]dt+\sigma dw(t),
$$
 with
	\[
		c=-\frac{\sigma^2}{\sigma^2/2+1-\sqrt{1+\sigma^4/4}}.
	\]
It follows that $m_t=\E\left\{x(t)\right\}$ satisfies
$$\dot{m}_t=\frac{\sigma^2}{\sigma^2t+c}m_t+\frac{c}{\sigma^2t+c},\quad m(0)=0, \quad m(1)=1.
$$
We get $m_t=t$.
The current drift of the Schr\"{o}dinger bridge is
$$\tilde{v}(x,t)=\frac{\sigma^2}{\sigma^2t+c}x+\frac{c}{\sigma^2t+c}-\frac{\sigma^2}{2}\nabla\log\tilde{\rho}_t(x),
$$
where $\tilde{\rho}_t$ has the form
$$\tilde{\rho}_t=(2\pi)^{-1/2}\exp\left[-\frac{(x-t)^2}{2q(t)}\right].
$$
Hence,
$$\tilde{v}(x,t)=\frac{\sigma^2}{\sigma^2t+c}x+\frac{c}{\sigma^2t+c}-\frac{\sigma^2}{2}\frac{t-x}{q(t)}.
$$
{ As $\sigma^2\searrow 0$, $c\rightarrow -2$, $q(t)\rightarrow 1$} and $\tilde{v}(x,t)\rightarrow 1, \forall x, \forall t$ (while $\tilde{u}(x,t)=\frac{\sigma^2}{2}\nabla\log\tilde{\rho}(x,t)\rightarrow 0$), which is just the optimal control   of the corresponding optimal transport problem. This is in agreement with the general theory \cite{Mik,leo}.

\section{Numerical example}\label{NE}
{%\mike
We consider  highly overdamped Brownian motion in a force field. Then, in a very strong sense \cite[Theorem 10.1]{N1}, the Smoluchowski model in configuration variables is a good approximation of the full Ornstein-Uhlenbeck model in phase space. We are interested in planar Brownian motion in the quadratic potential 
$$V(x)=\frac{1}{2}x'3I_2x=\frac{1}{2}[x_1,x_2]\left[\begin{matrix}3 & 0\\ 0 & 3\end{matrix}\right]\left[\begin{matrix}x_1\\ x_2\end{matrix}\right].
$$ 
Taking the mass of the particle to be equal to one, the planar evolution of the Brownian particle is given by the Smoluchowski equation
\begin{equation}\label{SMOL}
	dx(t)=-\nabla V(x(t))dt+\sqrt{\epsilon}dw(t), \quad -\nabla V(x)= Ax, \quad A=\left[\begin{matrix}-3 & 0\\ 0 & -3\end{matrix}\right],
\end{equation}
where $w$ is a standard, two-dimensional Wiener process.
} The observed distributions of the particle at the two end-points in time are normal with mean and variance
	\[
		m_0=\left[\begin{matrix}-5 \\-5\end{matrix}\right],\mbox{ and } \Sigma_0=
						\left[\begin{matrix}1 & 0\\ 0 & 1\end{matrix}\right]
	\]
at $t=0$, and
	\[
		m_1=\left[\begin{matrix}5 \\5\end{matrix}\right],\mbox{ and }
						\Sigma_1=\left[\begin{matrix}1 & 0\\ 0 & 1\end{matrix}\right]
	\]
at $t=1$, respectively. We then seek to interpolate the density of the particle at intermediate points by solving the corresponding Schr\"{o}dinger bridge problem where (\ref{SMOL}) plays the role of an a priori evolution.

Figure \ref{fig:interpolation1} depicts the flow between the two one-time marginals for the Schr\"{o}dinger bridge when $\epsilon=9$. The transparent tube represent the ``$3\sigma$ region''
	\[
		(x'-m_t')\Sigma_t ^{-1}(x-m_t) \le 9.
	\]
Typical sample paths are shown in the figure. Similarly, Figures \ref{fig:interpolation2} and  \ref{fig:interpolation3} depict the corresponding flows for $\epsilon=4$ and $\epsilon=0.01$, respectively. Figure \ref{fig:interpolation4} is the limit that represents optimal mass transport with prior velocity field $v(x,t)=Ax$; the sample paths are smooth curves that follow optimal transportation paths. As $\epsilon\searrow 0$, the paths {%\mike
of the bridge process resemble those of the corresponding optimal transport process} for $\epsilon=0$.
%This is consistent with the claim that the optimal transport problem is the limit of Schr\"{o}dinger bridges problem as $\epsilon$ goes to $0$.
For comparison, we also provide in Figure \ref{fig:interpolation5} the interpolation corresponding to optimal transport without a prior, which is given by a constant speed translation.
\begin{figure}\begin{center}
    \includegraphics[width=0.60\textwidth]{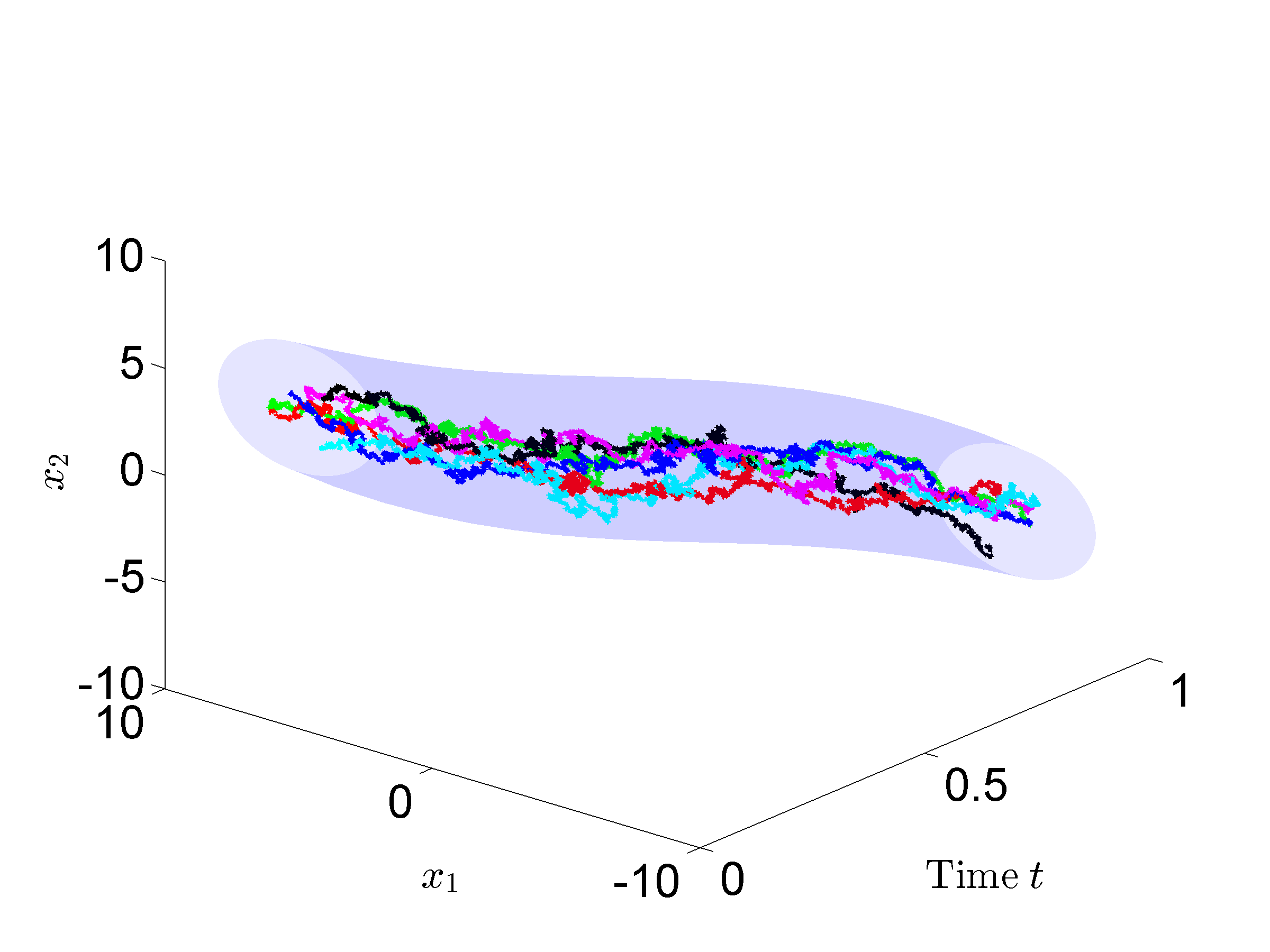}
    \caption{Interpolation based on Schr\"{o}dinger bridge with $\epsilon=9$}
    \label{fig:interpolation1}
\end{center}\end{figure}
\begin{figure}\begin{center}
    \includegraphics[width=0.60\textwidth]{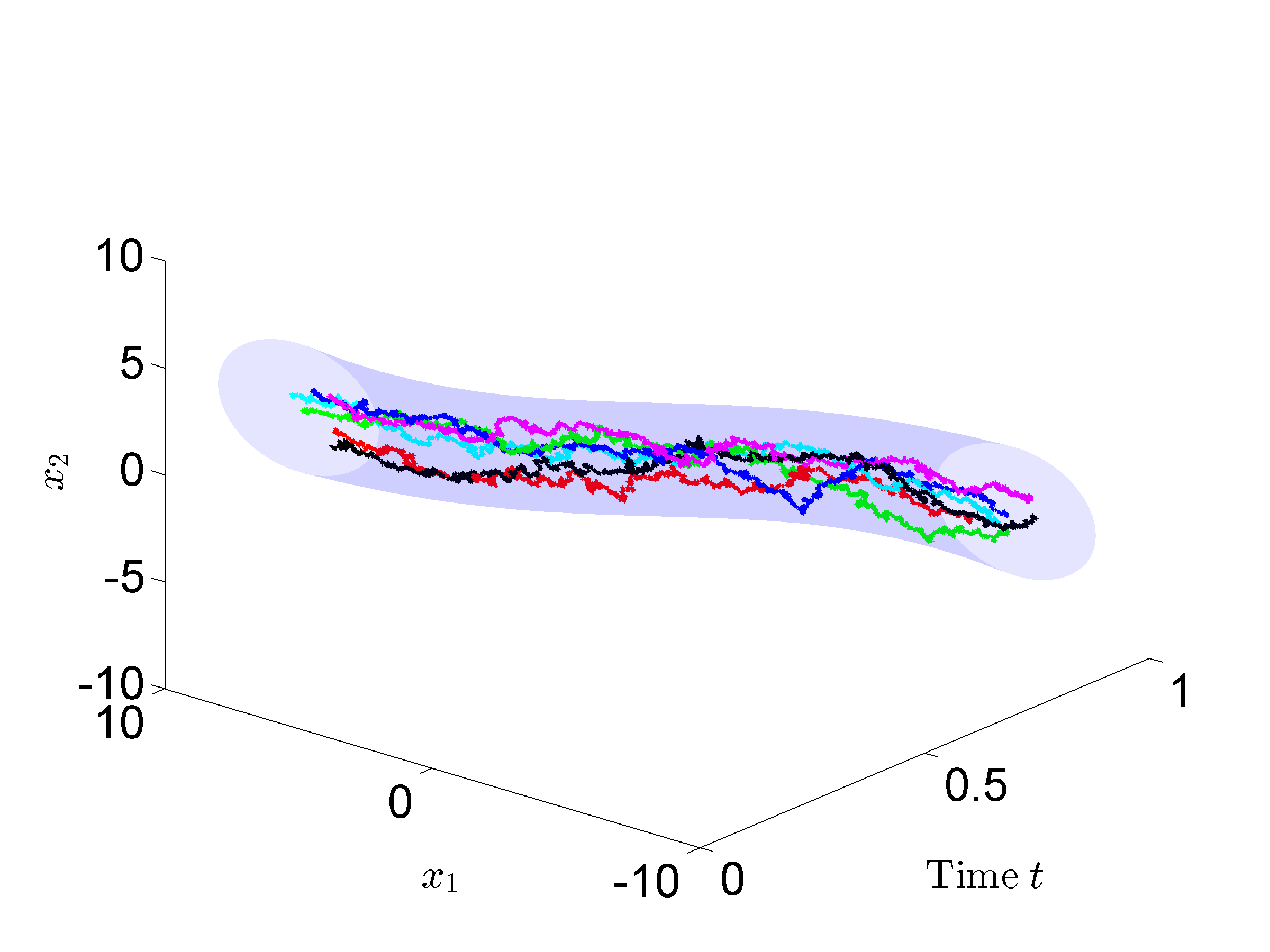}
    \caption{Interpolation based on Schr\"{o}dinger bridge with $\epsilon=4$}
    \label{fig:interpolation2}
\end{center}\end{figure}
\begin{figure}\begin{center}
    \includegraphics[width=0.60\textwidth]{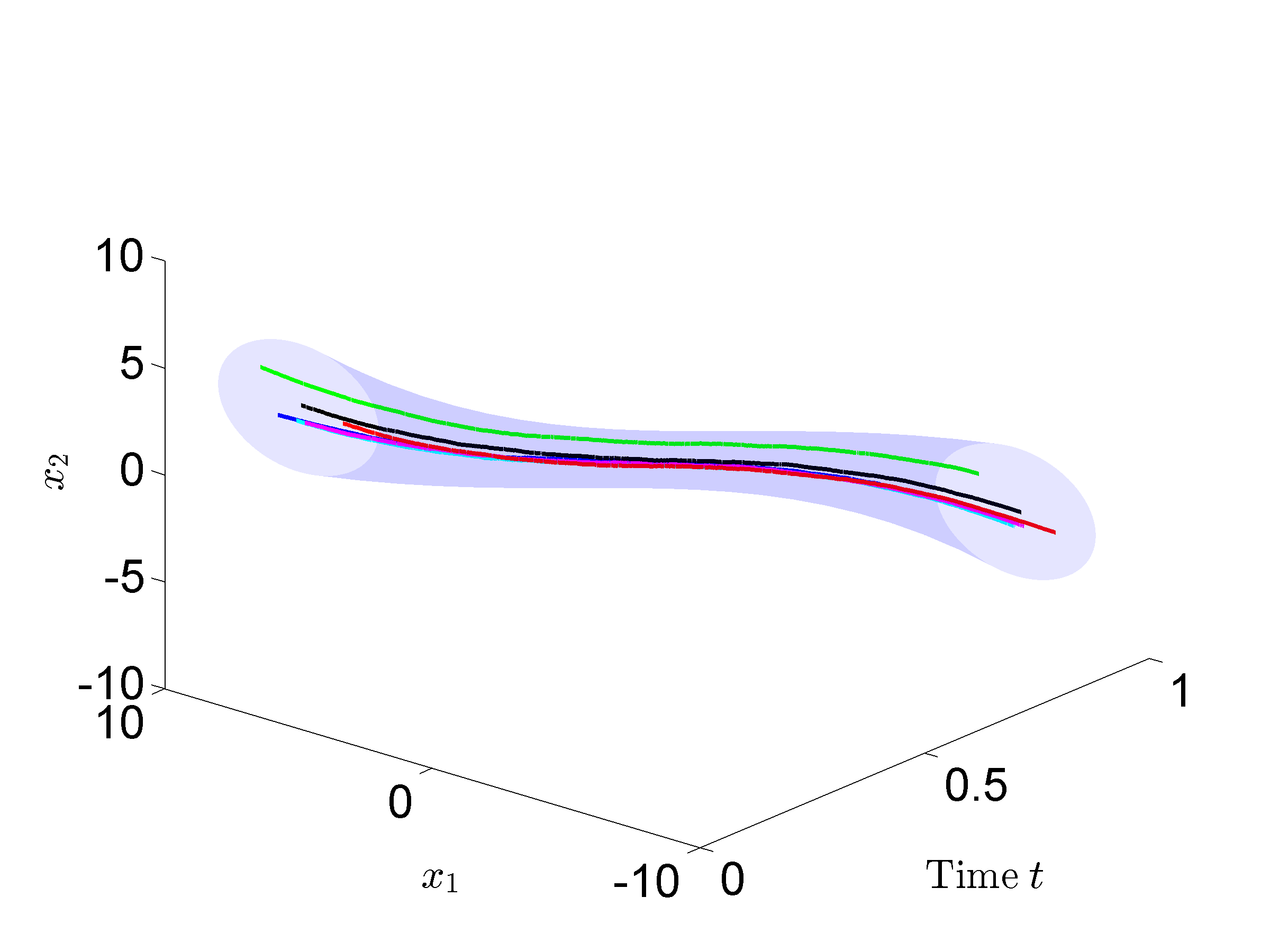}
    \caption{Interpolation based on Schr\"{o}dinger bridge with $\epsilon=0.01$}
    \label{fig:interpolation3}
\end{center}\end{figure}
\begin{figure}\begin{center}
    \includegraphics[width=0.60\textwidth]{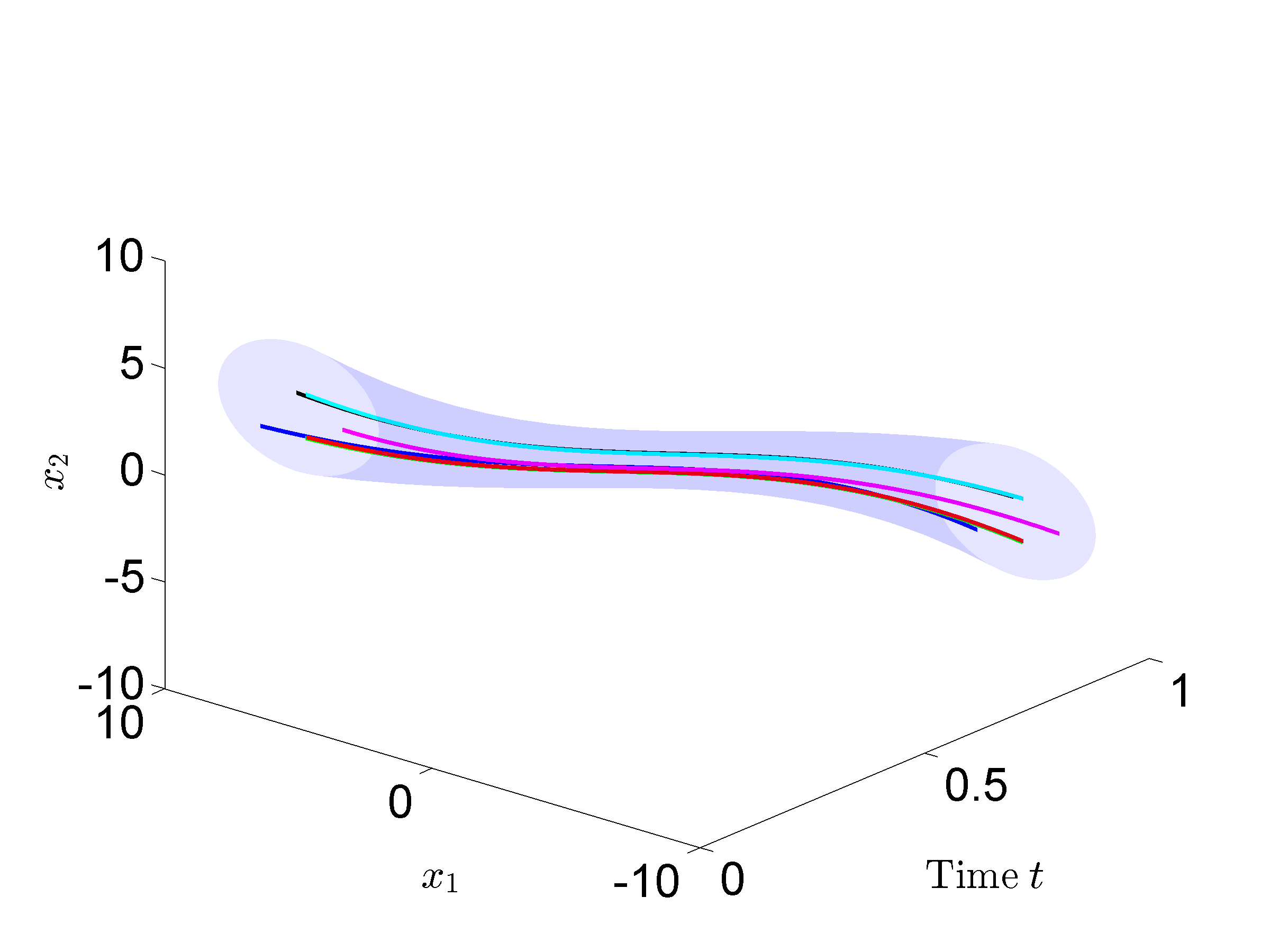}
    \caption{Interpolation based on optimal transport with prior}
    \label{fig:interpolation4}
\end{center}\end{figure}
\begin{figure}\begin{center}
    \includegraphics[width=0.60\textwidth]{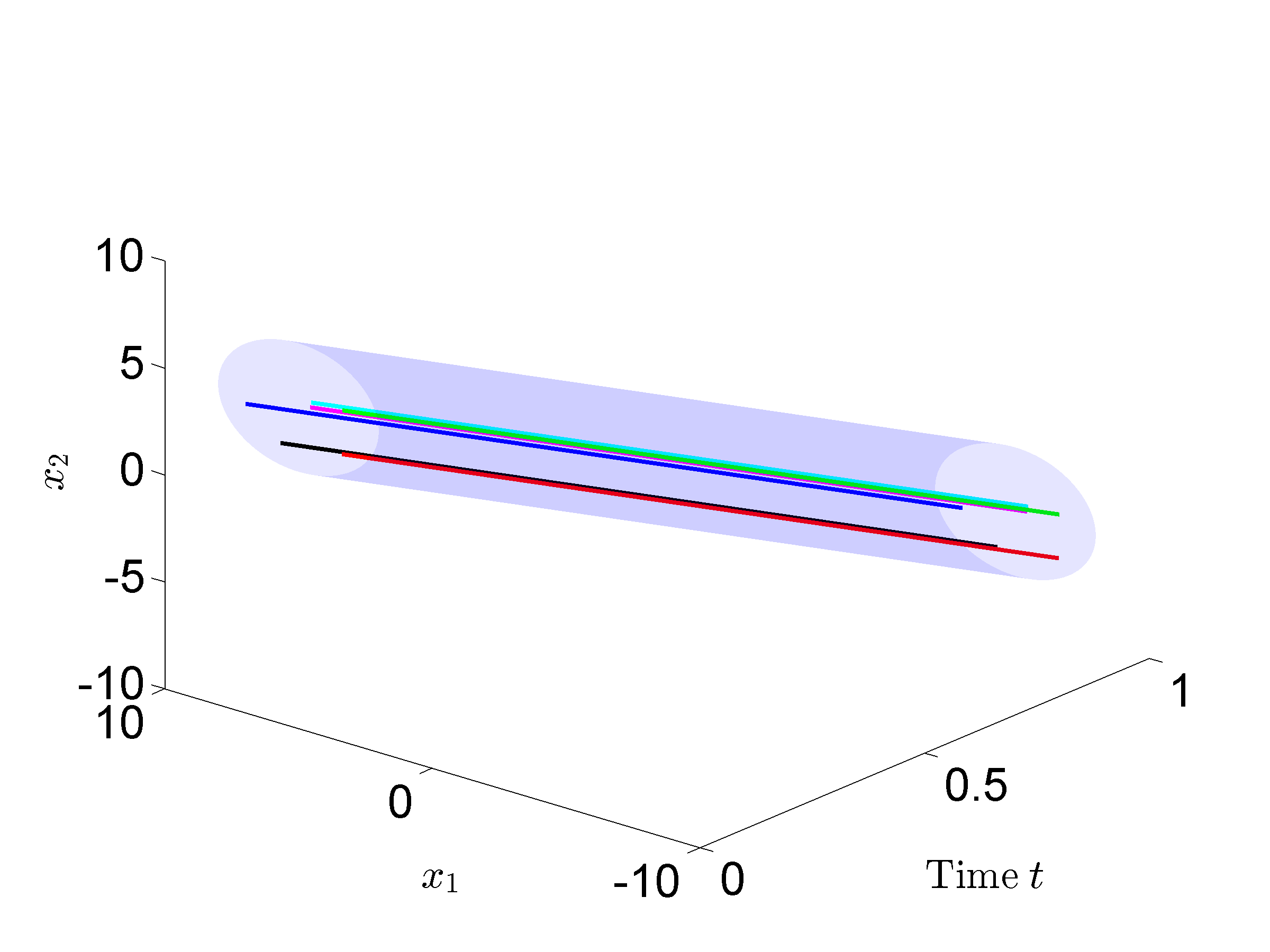}
    \caption{Interpolation based on optimal transport without prior}
    \label{fig:interpolation5}
\end{center}\end{figure}

\section*{Acknowledgment}
Part of the research of M.P. was conducted during a stay at the Courant Institute of Mathematical Sciences of the New York University whose hospitality is gratefully acknowledged.

\end{document}